\documentclass[sn-basic]{sn-jnl}

\usepackage{graphicx}%
\usepackage{multirow}%
\usepackage{amsmath,amssymb,amsfonts}%
\usepackage{amsthm}%
\usepackage{mathrsfs}%
\usepackage[title]{appendix}%
\usepackage{xcolor}%
\usepackage{textcomp}%
\usepackage{manyfoot}%
\usepackage{booktabs}%
\usepackage{algorithm}%
\usepackage{algorithmicx}%
\usepackage{algpseudocode}%
\usepackage{listings}%

\newtheorem{theorem}{Theorem}
\newcommand{\R}{\mathbb R}

\newcommand{\C}{\mathcal C}

\newcommand{\f}{\mathbf f}
\newcommand{\g}{\mathbf g}
\newcommand{\h}{\mathbf h}

\newcommand{\vb}{\mathbf v}
\newcommand{\w}{\mathbf w}
\newcommand{\x}{\mathbf x}
\newcommand{\y}{\mathbf y}
\newcommand{\z}{\mathbf z}

\newcommand{\A}{\mathbf A}

\newcommand{\F}{\mathbf F}
\newcommand{\G}{\mathbf G}
\newcommand{\Hb}{\mathbf H}
\newcommand{\I}{\mathbf I}

\newcommand{\Lb}{\mathbf L}
\newcommand{\M}{\mathbf M}
\newcommand{\N}{\mathbf N}
\newcommand{\Pb}{\mathbf P}
\newcommand{\Q}{\mathbf Q}

\newcommand{\0}{\mathbf 0}

\newcommand{\tr}{\mathrm{tr}}

\newcommand{\rank}{\mathrm{rank}}

\newcommand\Ss{\mathcal{S}}

\begin{document}

\title[Article Title]{A Randomized Exchange Algorithm for Optimal Design of Multi-Response Experiments}


\author*[1]{\fnm{P\'{a}l} \sur{Somogyi}}\email{pal.somogyi@fmph.uniba.sk}

\author[1]{\fnm{Samuel} \sur{Rosa}}\email{samuel.rosa@fmph.uniba.sk}

\author[1]{\fnm{Radoslav} \sur{Harman}}\email{radoslav.harman@fmph.uniba.sk}

\affil[1]{\orgdiv{Faculty of Mathematics, Physics and Informatics}, \orgname{Comenius University}, \orgaddress{\city{Bratislava}, \country{Slovakia}}}


\abstract{Despite the increasing prevalence of vector observations, computation of optimal experimental design for multi-response models has received limited attention. To address this problem within the framework of approximate designs, we introduce mREX, an algorithm that generalizes the randomized exchange algorithm REX (J Am Stat Assoc 115:529, 2020), originally specialized for single-response models. The mREX algorithm incorporates several improvements: a novel method for computing efficient sparse initial designs, an extension to all differentiable Kiefer's optimality criteria, and an efficient method for performing optimal exchanges of weights. For the most commonly used D-optimality criterion, we propose a technique for optimal weight exchanges based on the characteristic matrix polynomial. The mREX algorithm is applicable to linear, nonlinear, and generalized linear models, and scales well to large problems. It typically converges to optimal designs faster than available alternative methods, although it does not require advanced mathematical programming solvers. We demonstrate the usefulness of mREX on bivariate dose-response Emax models for clinical trials, both without and with the inclusion of covariates.}

\keywords{Optimal experimental design, Multi-response models, Convex optimization algorithms, D-optimality, Kiefer's criteria, Emax models\footnote[2]{The research was supported by the Collegium Talentum Programme of Hungary and the Slovak Scientific Grant Agency (grant VEGA 1/0327/25).}}



\maketitle

\section{Introduction}\label{sec1}

The aim of optimal experimental design is to conduct trials in a way that maximizes the information gained from an experiment, enabling the most accurate estimation of the unknown parameters of an underlying statistical model. For a general introduction to optimal design of experiments, see, for example, \cite{fedorov1972}, \cite{pazman86}, \cite{pukelsheim}, \cite{AtkinsonEA07}, \cite{fedorovleonov2013}, \cite{pronzatopazman2013} or \cite{lopezfidalgo2023}.

In this paper, we consider a design space $\mathcal{X}$ consisting of a finite number of design points $\x_1, \ldots, \x_N$, representing all permissible trial conditions. A finite design space is natural in many applications; moreover, if an application allows for a continuous design space, it can often be replaced with a fine discretization without significant loss in efficiency (e.g., \cite{gex}).

An approximate design $\w$ is any vector in $\R^N$ with non-negative components $w_1,\ldots,w_N$ that sum to $1$. For any $i$, the value of $w_i$ represents the proportion of trials to be performed at the design point $\x_i$. A usual representation of an approximate design $\w$ is by means of a table
\begin{equation}
\w =
\begin{pmatrix}
\x_{i_1} & \x_{i_2} & \cdots & \x_{i_K} \\
w_{i_1}  & w_{i_2}  & \cdots & w_{i_K} \\
\end{pmatrix}, 
\end{equation}
where $\x_{i_1}, \x_{i_2}, \ldots, \x_{i_K}$ form the support of $\w$, i.e., the set of all design points $\x_i$ such that $w_i>0$.

In practical experiments, the proportions provided by the approximate design need to be converted into actual numbers of trials, so-called exact designs, employing some rounding procedure (e.g., \cite{rounding}, cf. \cite{FilovaHarman20}). Since we focus on approximate designs, we will simply call them ``designs''.
\bigskip

For the standard linear regression model
\begin{equation}\label{eq:model_uni}
    y_\x = \f^T(\x) \boldsymbol{\beta} + \epsilon_\x,
\end{equation}
with real-valued responses $y_\x$, an unknown vector $\boldsymbol{\beta} \in \R^m$ of parameters and uncorrelated homoscedastic errors, the amount of information gained from the experiment designed according to $\w$ is expressed by the information matrix
\begin{equation}\label{eq:infmat_uni}
    \M(\w) = \sum_{i=1}^N w_i \f(\x_i) \f^T(\x_i).
\end{equation}
If $\epsilon_\x \sim \mathcal{N}(0, \sigma^2)$ and $\w$ corresponds to an exact design of size $n$, then $\sigma^{-2}n\M(\w)$ is the Fisher information matrix for $\boldsymbol{\beta}$. Note that the information matrix can also be expressed as
\begin{equation*}
    \M(\w) = \sum_{i=1}^N w_i \Hb(\x_i),
\end{equation*}
where $\Hb(\x_i) = \f(\x_i) \f^T(\x_i)$ is the ``elementary'' information matrix that measures the amount of information gained from one trial at the design point $\x_i \in \mathcal{X}$, for any $i$. For the standard model \eqref{eq:model_uni},  the ranks of the elementary information matrices are $1$. The combined information matrix $\M(\w)$ is clearly positive semidefinite, and can be of full rank $m$.

The quality of a design $\w$ is measured by the size of $\M(\w)$, via a function $\Phi$ known as optimality criterion, defined on the set of all positive semidefinite matrices. That is, the objective of optimal experimental design is to solve the optimization problem:
\begin{equation}\label{eq:genproblem}
    \mathop{\mathrm{argmax}}_{\w \in \Xi} \Phi[\M(\w)],
\end{equation}
where $\Xi = \{ \w \in \R^N \, : \, w_1 \geq 0, \ldots, w_N \geq 0, \sum_i w_i = 1 \}$ is the set of all designs; formally the probability simplex in $\R^N$. Any solution $\w^*$ to this optimization problem is referred to as a $\Phi$-optimal design. 

While there exist many optimality criteria, we focus on the class of Kiefer's $\Phi_p$-criteria for $p \in [0,\infty)$, also called matrix means (see, e.g., \cite{pukelsheim}, Section 6.7),  defined as:
\begin{equation}
    \Phi_p \left[ \M \right] =
\begin{cases}
(\det\left[\M\right])^{1/m}  & \text{ for } p = 0 \text{ and } \M \text{ nonsingular}, \\
\left( \frac{1}{m} \: \mathrm{tr}\left[\M^{-p}\right] \right)^{-1/p} & \text{ for } p \in \left(0, \infty \right) \text{ and } \M \text{ nonsingular},\\
0  & \text{ for } \M \text{ singular}. \label{kiefercriterion}
\end{cases}
\end{equation}
For $p=0$ and $p=1$, formula \eqref{kiefercriterion} corresponds to the most prominent criteria of $D$- and $A$-optimality, respectively. The limit of $\Phi_p$ as $p$ tends to infinity, gives the criterion of $E$-optimality, which is equal to the smallest eigenvalue of $\M$  (e.g., Section 6.4 in \cite{pukelsheim}). 
\bigskip

The properties of optimal designs for the single-response model \eqref{eq:model_uni} and their constructions, both analytical and algorithmic, are widely studied in the literature. In this paper, we instead focus on the optimal design problem for the linear regression model with multivariate responses:
\begin{equation}\label{eq:model_multi}
    \y_\x = \F^T(\x) \boldsymbol{\beta} + \boldsymbol{\epsilon}_\x, 
\end{equation}
which is much less studied from the optimal design perspective. Here, $\y_\x$ represents the $s$-dimensional response from the trial at the design point $\x \in \mathcal{X}$, $\boldsymbol{\beta} \in \mathbb{R}^m$ is the vector of parameters, and $\F(\x)$ is a known $m \times s$ matrix. The vectors $\boldsymbol{\epsilon}_\x$ of random errors are assumed to have expectation $\mathbf{0}_s$ and covariance matrix $\sigma^2\boldsymbol{\Sigma}$, where $\boldsymbol{\Sigma}$ is a known and positive definite $s \times s$ matrix. Note that the dimension of $\boldsymbol{\Sigma}$ is typically small (we often consider bivariate response models, i.e., $s = 2$), and it can be estimated from previous experiments (cf. \cite{ting2006}, \cite{atashgah2007}, \cite{magnusdottir2016}). Across different trials, the error vectors are assumed to be independent. 

In short, we refer to \eqref{eq:model_uni} and \eqref{eq:model_multi} as the single-response and multi-response regression (or model), respectively. For simplicity, we formulate our results for the linear multi-response settings \eqref{eq:model_multi}, but note that they easily extend to nonlinear models, as shown in Section \ref{sec:further}.

For the multi-response regression, the appropriate information matrix of a design $\w$ is 
\begin{equation*}
    \M(\w) = \sum_{i=1}^N w_i \F(\x_i) \boldsymbol{\Sigma}^{-1} \F^T(\x_i) = \sum_{i=1}^N w_i \Hb(\x_i),
\end{equation*}
where the elementary information matrix corresponding to the design point $\x_i \in \mathcal{X}$ is of the form 
\begin{equation*}
        \Hb(\x_i) =  \G(\x_i)\textbf{G}^T(\x_i), \G(\x_i) = \F(\x_i)\boldsymbol{\Sigma}^{-1/2}.
\end{equation*}
 Note that, for any $i$, the ranks of the matrices $\F(\x_i)$, $\G(\x_i)$ and $\Hb(\x_i)$ are the same, and typically greater than one. As in the single-response regression, a $\Phi$-optimal design for multi-response regression is a solution to \eqref{eq:genproblem}. To make the $\Phi_p$-optimal design problem nontrivial, we assume that there exists a nonsingular design, i.e., a design $\w$ such that $\M(\w)$ is nonsingular. This assumption is equivalent to 
\begin{equation}\label{eq:basic}
    \C(\G(\x_1),\ldots,\G(\x_N))=\R^m.
\end{equation}

Although the $\Phi$-optimal design problem for multi-response models is less studied, it is important and appears in numerous modern applications (e.g., \cite{magnusdottir2016}, \cite{seufert2021},  \cite{seurat2021}, \cite{tsirpitzi_miller}, or \cite{radloff2023}). Nevertheless, only a few algorithms for computing $\Phi$-optimal designs for the multi-response model on a finite design space, typically for specific $\Phi_p$-criteria, have been considered in the literature. These include the multiplicative method (\cite{ucinski2007}, \cite{ht09}, \cite{yu2010}), an approach based on the Newton method (\cite{ybt}, cf. \cite{pronzato2014}) and the second-order cone programming approach (\cite{sagnol_socp}, \cite{sagnol2015}). We also note that some algorithms for single-response models, such as the vertex direction method (\cite{wynn1970}, \cite{fedorov1972}) and semidefinite programming (e.g., \cite{VandenbergheBoyd}), can be directly generalized for multi-response models (e.g., \cite{atashgah2007}, \cite{atashgah2009}, \cite{wong2018} or \cite{duarte}).

Building on the work of \cite{bohning}, a randomized exchange algorithm (REX; \cite{rex}) was recently proposed for the single-response model. This algorithm rapidly converges even for problems with relatively large sizes $N$ of the design space and numbers $m$ of parameters, and outperforms the competing methods in many situations (as demonstrated in Section 4 of \cite{rex}). In this paper, we propose the mREX (multi-response REX) algorithm, which generalizes REX to provide a fast and reliable solution for the multi-response settings.

The novel aspects of the generalization pertain in particular to:
\begin{itemize}
    \item \textbf{A method of computing initial designs.} As we prove, the design obtained by the proposed initialization method always has at most $m$ non-zero components and a  nonsingular information matrix. The design is typically much more efficient than a design generated uniformly randomly. The simultaneous properties of ``sparsity'' and efficiency of the initial design make it particularly suitable for algorithms such as mREX.
    \item \textbf{Extensions to involve all Kiefer's $\Phi_p$-criteria with $p \in [0,\infty)$.} The original REX algorithm was formulated only for $D$- and $A$-optimality, corresponding to Kiefer's criteria for $p=0$ and $p=1$, respectively. This extension allows experimenters to consider, for instance, criteria with $p \in (0,1)$ which form a compromise between $D$- and $A$-optimality, or criteria with larger $p$, which are close to the $E$-optimality criterion.
    \item \textbf{A method of solving inner optimization problems.} Since in multi-response problems we in general do not have analytical formulas for the solution of key inner optimization problems, we propose an efficient numerical scheme. This scheme is based on the idea of fast ``full'' exchange verification and, for the most common criterion of $D$-optimality, it is further significantly improved by a method based on the characteristic polynomial of a matrix.
\end{itemize}

We will show that the mREX algorithm retains the beneficial properties of the standard REX algorithm; in particular, it is simple to implement in any computing environment, can be applied to large problems, and generally outperforms the state-of-the-art competing methods applicable to the optimal design problem of multi-response regression.

The rest of this paper is organized as follows. In Section \ref{sec:problem}, we summarize the fundamental mathematical properties of the multi-response optimal design problem for the Kiefer's optimality criteria, such as the lower bound on the efficiency of any given design. The central Section \ref{sec:mREX} presents the algorithm mREX, particularly the computation of initial designs and solution of the inner optimization problems. Section \ref{sec:further} examines extensions of the model \eqref{eq:model_multi} that allow for the use of mREX, and Section \ref{sec:examples} provides numerical studies demonstrating the utility of mREX for an assessment of a dose-response design. In this section, we additionally demonstrate the computational advantages of mREX compared to existing algorithms capable of solving the multi-response optimal design problem. Proofs of mathematical theorems are given in Appendix.
\bigskip

\textbf{Notation.} The symbols $\Ss^m$, $\Ss^m_+$ and $\Ss^m_{++}$ denote the sets of all $m \times m$ symmetric, (symmetric) positive semidefinite and (symmetric) positive definite matrices, respectively. The symbols $\0_{m \times m}$ and $\0_m$ denote the $m\times m$ matrix of zeros and the $m \times 1$ vector of zeros, respectively. For a matrix $\A$, $\C(\A)$ is the column space of $\A$, $\A^+$ is the Moore-Penrose pseudoinverse of $\A$, and $(\A)_{ij}$ is the $ij$-th element of $\A$.
For a linear subspace $\mathcal{L}$ of $\R^m$, we denote the orthogonal projector onto $\mathcal{L}$ by $\Pb_\mathcal{L}$, and $\mathcal{L}^\perp$ denotes the orthogonal complement of $\mathcal{L}$. For a real number, the symbol $\lceil.\rceil$ denotes the ceiling function.

\section{The multi-response problem of Kiefer's optimality}\label{sec:problem}

In this section, we provide some fundamental properties of the multi-response problem of Kiefer's $\Phi_p$-optimality. For any $p \in [0, \infty)$, the criterion $\Phi_p$ is positively homogeneous, continuous, and concave on $\Ss^m_+$, strictly positive on $\Ss^m_{++}$, and it vanishes on $\Ss^m_+\setminus\Ss^m_{++}$; cf. Chapter 6 in \cite{pukelsheim}. The criterion $\Phi_p$ is not strictly concave on $\Ss^m_{++}$ in the traditional sense.\footnote{$\Phi_p$ \emph{is} strictly concave in the nontypical sense of definition in \cite{pukelsheim}, Section 5.2.} However, a simple extension of the argument in Section 6.17 of \cite{pukelsheim} can be used to prove that $\Phi_p$ is strictly \emph{log}-concave on $\Ss^m_{++}$. 

Note that the set of all information matrices $\mathcal{M}:=\{\M(\w): \w \in \Xi\}$ is the convex hull of the finite set $\{\Hb(\x_1),\ldots,\Hb(\x_N)\}$, therefore it is convex and compact. Moreover, the set $\mathcal{M} \cap \Ss^m_{++}$ of all nonsingular information matrices is convex and nonempty due to our assumption of the existence of a nonsingular design. Thus, the properties of the Kiefer's criteria imply that there exists at least one $\Phi_p$-optimal design, and the information matrix of any $\Phi_p$-optimal design is nonsingular. The strict log-concavity of $\Phi_p$ in addition implies that the information matrix of all $\Phi_p$-optimal designs is unique.

While the $\Phi_p$-optimal information matrix is unique, there is often a continuum of $\Phi_p$-optimal designs. Analogously to the single-response $\Phi_p$-optimal design problem, it is possible to show (cf. Section 8.2 in \cite{pukelsheim}) that there always exists a $\Phi_p$-optimal design with a support size of at most $\frac{1}{2}m(m+1)+1$. Hence, since in statistical applications we usually have $m \ll N$, there typically exists a ``sparse'' $\Phi_p$-optimal design. Importantly, while in the single-response case the size of the support of an optimal design is at least $m$, in the multi-response problem the optimal support size is often smaller than $m$. Clearly, a lower bound on the support size of any $\Phi_p$-optimal design is the minimum $k$ such that the matrix $(\G(\x_{i_1}),\ldots,\G(\x_{i_k}))$ has full row rank for some indices $i_1,\ldots,i_k$, i.e., the support of a $\Phi_p$-optimal design can be as small as $\lceil m/s \rceil$. 
\bigskip

The $\Phi_p$-efficiency of a design $\w$ relative to a nonsingular design $\widetilde{\w}$ is $\Phi_p\left[ \M(\w) \right] / \Phi_p\left[ \M(\widetilde{\w}) \right]$. We refer to the $\Phi_p$-efficiency of $\w$ relative to a $\Phi_p$-optimal design $\w^*$ simply as $\Phi_p$-efficiency of $\w$ (e.g., Section 5.15 in \cite{pukelsheim}): 
\begin{equation*}
    \text{eff}_{\Phi_p} (\w) = \frac{\Phi_p\left[ \M(\w) \right]}{\Phi_p\left[ \M(\w^*) \right]}.
\end{equation*}
Because the criteria $\Phi_p$ are positively homogeneous, $\text{eff}_{\Phi_p} (\w)$ has a clear interpretation: Assume that we are able to design an experiment of size $n$ according to an approximate design $\w$. Then, any exact design attaining the same or a better value of the $\Phi_p$-criterion must include at least $n \: \text{eff}_{\Phi_p} (\w)$ trials. Importantly, the convex nature of the problem \eqref{eq:genproblem} with $\Phi=\Phi_p$ provides means of computing a lower bound on the $\Phi_p$-efficiency of any given design $\w$, as we explain next.

The criterion $\Phi_p$ is smooth on $\mathcal{S}^m_{++}$ with gradient (e.g., \cite{pukelsheim}, Section 7.19)
\begin{equation}
 \nabla \Phi_p \left[ \M \right] = 
   \frac{ \Phi_p \left[ \M \right]}{\text{tr} \left[\M^{-p} \right] } \M^{-p-1} \quad \text{ for all } \M \in \mathcal{S}_{++}^m.
\label{nablakiefer}
\end{equation}

A tool often utilized for the theory and computation of $\Phi_p$-optimal designs is the directional derivative (e.g., \cite{pazman86}, Section VI.4). The directional derivative of $\Phi_p$ in $\M \in \mathcal{S}_{++}^m$ in the direction of $\mathbf{N} \in \mathcal{S}_+^m$ is defined as
\begin{equation}
    \partial \Phi_p \left[\M, \mathbf{N}\right] = \lim_{\beta \rightarrow 0^+} 
    \frac{\Phi_p\left[(1-\beta)\M + \beta \mathbf{N} \right] - \Phi_p\left[\M \right]}{\beta}
\end{equation}
and it can be explicitly calculated as follows (e.g., \cite{gaffke1985}):
\begin{equation}\label{eq:dirderPhip}
    \partial \Phi_p \left[\M, \mathbf{N}\right] = \text{tr} \left[ \nabla \Phi_p\left[\M\right](\mathbf{N}-\M) \right]= \Phi_p \left[\M\right] \left( \frac{\text{tr}\left[\M^{-p-1}\mathbf{N}\right]}{\text{tr}\left[\M^{-p}\right]} -1 \right).
\end{equation}
Using the directional derivatives, we provide the following bound on the $\Phi_p$-efficiency of a design.
\bigskip
\begin{theorem}\label{thm1}
    The $\Phi_p$-efficiency of a nonsingular design $\w$ satisfies
    \begin{equation}\label{eq:effbnd}
        \mathrm{eff}_{\Phi_p}(\w) \geq \frac{\mathrm{tr}\left[\M^{-p}(\w)\right]}{\max_{i=1,\ldots,N} (\mathbf{g}_p(\w))_i},
    \end{equation}
    where $\mathbf{g}_p(\w)$ is a vector with components $i=1,\ldots,N$ given by
    \begin{eqnarray}\label{eq:g}
        (\mathbf{g}_p(\w))_i&=&\mathrm{tr}\left[\M^{-p-1}(\w)\Hb(\x_i)\right] \nonumber \\
        &=& \mathrm{tr}\left[\G^T(\x_i)\M^{-p-1}(\w)\G(\x_i)\right].
    \end{eqnarray}
\end{theorem}

Using the directional derivatives, it is also straightforward to prove the so-called equivalence theorem for the multi-response optimal design problem (cf. \cite{kiefer1974}, Section 5): A nonsingular design $\w$ is optimal if and only if
\begin{equation}
\max_{i=1,\ldots,N} (\mathbf{g}_p(\w))_i  = \text{tr}\left[\M^{-p}(\w)\right].\label{zastavvetaoekv}
\end{equation}

Note that \eqref{eq:effbnd} aligns with its single-response counterpart used in the REX algorithm: In \cite{rex}, the $D$- and $A$-efficiency of a nonsingular design $\w$ is bounded by
\begin{eqnarray}
    \mathrm{eff}_D(\w) \geq \frac{m}{\max_{i=1,\dots,N} (\g(\w))_i}, \label{eq:effbnd_single_D} \\
    \mathrm{eff}_A(\w) \geq \frac{\text{tr}\left( \M^{-1}(\w) \right)}{\max_{i=1,\dots,N} (\g(\w))_i}, \label{eq:effbnd_single_A}
\end{eqnarray}
where 
\begin{equation}
    (\g(\w))_i = \begin{cases}
        \f^T(\x_i) \M^{-1}(\w) \f(\x_i) \quad \text{ for $D$-optimality,} \\
        \f^T(\x_i) \M^{-2}(\w) \f(\x_i) \quad \text{ for $A$-optimality.}
    \end{cases} \label{eq:gwDA}
\end{equation}

Various theoretical properties of the multi-response design problem for particular $\Phi_p$-criteria are given by \cite{ht09} for $D$-optimality, and by \cite{wong2018} for $D$- and $A$-optimality.

\section{The randomized exchange algorithm for multi-response optimality}\label{sec:mREX}

\begin{algorithm}[ht!]
\caption{Randomized exchange algorithm for single-response models (REX) adapted from \cite{rex}: The value of $\mathrm{eff.bnd}(\w)$ (Line 2) is the lower bound on the efficiency of the current design $\w$, as given by the right-hand side of \eqref{eq:effbnd_single_D} or \eqref{eq:effbnd_single_A}. The vector $\mathbf{g}(\w)$ (Line 4) is defined in equation \eqref{eq:gwDA} depending on whether we intend to compute $D$- or $A$-optimal designs.}
\label{REX}
\begin{algorithmic}[1]
    \Statex \textbf{Input:} The vectors $\f(\x_1),\ldots,\f(\x_N) \in \mathbb{R}^m$, threshold efficiency $\mathit{eff} \in (0,1)$, tuning parameter $\gamma \geq 1/m$.
    \Statex \textbf{Output:} Approximate design $\w$
    \Statex
    \State Generate a random regular design $\w$
    \While {$\mathrm{eff.bnd}(\w) < \mathit{eff}$}
        \State Let $(i^s_1,\ldots,i^s_K)$ be a random permutation of the $K$ indices of the support points of $\w$
        \State Let $(i^g_1,\ldots,i^g_L)$ be a random permutation of the indices of the $L = \min(\lceil \gamma m\rceil, N)$ largest elements of $\mathbf{g}(\w)$
        \For {$l \gets 1$ to $L$}
            \For {$k \gets 1$ to $K$}
                \State Set $\mathbf{\Delta} \leftarrow \f(\x_{i^g_l})\f^T(\x_{i^g_l})-\f(\x_{i^s_k})\f^T(\x_{i^s_k})$
                \State Find $\alpha^*$ in $\mathrm{argmax}\{\Phi[\M(\w)+\alpha\mathbf{\Delta}]): \alpha \in [-w_{i^g_l}, w_{i^s_k}]\}$
                \State Set $w_{i^g_l} \leftarrow w_{i^g_l} + \alpha^*$ ; $w_{i^s_k} \leftarrow w_{i^s_k} - \alpha^*$
            \EndFor
        \EndFor
    \EndWhile
\end{algorithmic}
\end{algorithm}

\begin{algorithm}[ht!]
\caption{Randomized exchange algorithm for multi-response models (mREX): The initial design (Line 1) is computed via the algorithm described in Section \ref{subsec:ini}. The value of $\mathrm{eff.bnd}_p(\w)$ (Line 2) is the lower bound on the efficiency of the current design $\w$, as given by the right-hand side of \eqref{eq:effbnd}. The vector $\mathbf{g}_p(\w)$ (Line 4) is defined in equation \eqref{eq:g}; see Section \ref{subsec:cand} for justification. The crucial inner optimization step (Line 8) is detailed in Section \ref{subsec:opt}.}
\label{mREX}
\begin{algorithmic}[1]
    \Statex \textbf{Input:} The $m \times s$ matrices $\G(\x_1),\ldots,\G(\x_N)$, parameter $p \in [0,\infty)$ determining the Kiefer's criterion, threshold efficiency $\mathit{eff} \in (0,1)$.
    \Statex \textbf{Output:} Approximate design $\w$
    \Statex
    \State Compute a nonsingular initial design $\w$
    \While {$\mathrm{eff.bnd}_p(\w) < \mathit{eff}$}
        \State Let $(i^s_1,\ldots,i^s_K)$ be a random permutation of the $K$ indices of the support points of $\w$
        \State Let $(i^g_1,\ldots,i^g_L)$ be a random permutation of the indices of the $L=\min(m,N)$ largest elements of $\mathbf{g}_p(\w)$
        \For {$l \gets 1$ to $L$}
            \For {$k \gets 1$ to $K$}
                \State Set $\mathbf{\Delta} \leftarrow \G(\x_{i^g_l})\G^T(\x_{i^g_l})-\G(\x_{i^s_k})\G^T(\x_{i^s_k})$
                \State Find $\alpha^*$ in $\mathrm{argmax}\{\Phi_p[\M(\w)+\alpha\mathbf{\Delta}]): \alpha \in [-w_{i^g_l}, w_{i^s_k}]\}$
                \State Set $w_{i^g_l} \leftarrow w_{i^g_l} + \alpha^*$ ; $w_{i^s_k} \leftarrow w_{i^s_k} - \alpha^*$
            \EndFor
        \EndFor
    \EndWhile
\end{algorithmic}
\end{algorithm}

In this section, we detail our randomized exchange algorithm mREX. The original REX algorithm for a single-response model (Algorithm \ref{REX}) repeatedly updates the current design $\w$ by performing optimal (full or partial) exchanges of weights for pairs of points: one point $\x_{i^s}$ at which $\w$ is currently supported, and the other point $\x_{i^g}$ at which the so-called variance function is large, as such points are more promising to support an optimal design. In each iteration of the algorithm REX, the points $\x_{i^s}$ and $\x_{i^g}$ are selected in batches, and the sequence of pairings for optimal exchanges is randomized. 

For our algorithm mREX (Algorithm \ref{mREX}), we generalize these steps to the multi-response settings, adding several new improvements. First, the original REX is initialized with a uniformly random $m$-point design (Line 1 of Algorithm \ref{REX}), possibly repeatedly, until the randomly generated design is nonsingular. For mREX, we provide a more efficient, nontrivial initialization (Line 1 of Algorithm \ref{mREX}, detailed in Section \ref{subsec:ini}), the importance of which is even more pronounced for the multi-response optimal design problem, as there an optimal design is typically supported on fewer points than $m$. Second, we provide a method for selecting the design points for the exchanges (Line 4 of Algorithm \ref{mREX}) that covers all Kiefer's $\Phi_p$-criteria, not just $D$- and $A$-optimality as in REX. 
Third, while the optimal exchanges in REX (Line 8 of Algorithm \ref{REX}) are determined using analytic formulas, an explicit analytic solution is not available in the multi-response setting. Therefore, for mREX, we propose a computationally effective numerical approach (Line 8 of Algorithm \ref{mREX}, detailed in Section \ref{subsec:opt}). Specifically, for the criterion of $D$-optimality, we reformulate the computation of optimal exchanges as finding the maximum of a fixed polynomial, which can be computed particularly efficiently.
Note that the selection of candidate points, as well as the randomized-exchange approach itself, is generalized in a straightforward manner, while the approaches to computing the initial design and performing optimal exchanges between design points represent nontrivial extensions of the original REX algorithm.

In the original REX algorithm, a tuning parameter $\gamma \geq 1/m$ controls the number of candidate points to exchange. More precisely, $L = \min(\lceil \gamma m\rceil, N)$ points are selected. We simplify the REX algorithm by setting $\gamma = 1$ and additionally omitting the so-called leading Böhning exchange, which was required for the convergence proof of the REX algorithm. Our numerical experience suggests that even this simplified version converges rapidly to a $\Phi_p$-optimal design.

\subsection{Initial design}\label{subsec:ini}

A strategic choice of initial design can significantly improve the overall performance of an optimal design algorithm. This becomes even more pronounced in the multi-response problem, where the size of the support of an efficient initial design can be much smaller than $m$, leading to greatly improved computational speed. However, to our knowledge, the only advanced approach to the initialization of optimal design algorithms for multi-response models is given by \cite{chang1997}, where the initial design is formed as the uniform design on the union of supports of optimal designs for single-response models. The initialization method by \cite{chang1997} requires solving multiple optimal design problems and may lead to unnecessarily large supports. Other algorithms either use initial designs supported at $m$ random points (e.g., \cite{ybt}), or are uniform on the entire design space (e.g., \cite{yu2010}).

In contrast, more work has been done on the construction of efficient initial designs for single-response models, with state-of-the-art methods being greedy heuristics (see \cite{GalilKiefer} and \cite{PIN}, cf. Section 12.6 in \cite{AtkinsonEA07}). These methods show good theoretical and practical performance; we therefore propose an initialization method for multi-response models inspired by the single-response approaches: our method, mKYM, is a multi-response generalization of the Kumar-Y{\i}ld{\i}r{\i}m algorithm (KYM) from \cite{PIN}, cf.\ \cite{KumarYildirim}.
\bigskip

The original KYM first constructs an $m$-element set $\Ss$ of design point indices as follows: it starts with $\Ss=\emptyset$ and, in each of $m$ iterations, adds the index
\begin{equation}\label{eq:KYMmax}
    i^* \in \arg\max_{i\not\in\Ss} \vert \vb^T \f(\x_i) \vert
\end{equation}
to $\Ss$. In equation \eqref{eq:KYMmax}, $\vb$ is a random direction in $\C^\perp(\M(\Ss))$, where $\M(\emptyset) = \0_{m \times m}$ and $\M(\Ss)=\sum_{i \in \Ss} \f(\x_i)\f^T(\x_i)$ for $\Ss \neq \emptyset$. Geometrically, this corresponds to selecting a design point $\x_{i^*}$ such that $\f(\x_{i^*})$ is largest in a random direction in the subspace orthogonal to the already selected $\{\f(\x_i): i \in \Ss\}$. 
The final output of the original KYM is a uniform design on $\{\x_i: i \in \Ss\}$; see \cite{PIN} for details.
\bigskip

 We propose the multi-response version mKYM as follows; cf. Algorithm \ref{aINI}. The method starts with $\Ss=\emptyset$ and, in each iteration, adds
 \begin{equation}\label{eq:mKYMmax}
     i^* \in \arg\max_{i\not\in\Ss} \Vert \vb^T \G(\x_i) \Vert
 \end{equation}
 to $\Ss$. In equation \eqref{eq:mKYMmax}, $\vb$ is a random direction in $\C^\perp(\M(\Ss))$, where $\M(\emptyset) = \0_{m \times m}$, $\M(\Ss)=\sum_{i \in \Ss} \G(\x_i)\G^T(\x_i)$ for $\Ss \neq \emptyset$, and $\lVert \cdot \rVert$ is the usual Euclidean norm. This can be implemented by iteratively updating the orthogonal projector $\Pb$ onto $\C^\perp(\M(\Ss))$. The algorithm terminates once the column space of $\M(\Ss)$ captures all $m$-dimensional vectors. This is equivalent to $\C^\perp(\M(\Ss)) = \{\0_m\}$, i.e., $\tr\left[\Pb\right] = \rank\left[\Pb\right] = 0$, which can be verified in a numerically stable manner as $\tr\left[\Pb\right] < 0.5$. The final output of the mKYM is the design uniform on $\{\x_i: i \in \Ss\}$. The following theorem proves that the resulting design is indeed nonsingular and has support of size at most $m$.
\bigskip


\begin{algorithm}
\caption{The multi-response Kumar-Y{\i}ld{\i}r{\i}m initiation (mKYM): see the description in Section \ref{subsec:ini} for details.}
\label{aINI}
\begin{algorithmic}[1]
    \Statex \textbf{Input:} The $m \times s$ matrices $\G(\x_1),\ldots,\G(\x_N)$
    \Statex \textbf{Output:} Approximate design $\w$
    \Statex
    \State Set $\Ss \gets \emptyset$ and $\Pb \gets \I_m$
    \Repeat
        \State Choose $\vb \in \C(\Pb)$ such that $\vb \neq \0_m$
        \State Choose $i^*$ in $\mathrm{argmax}_{i \in \{1,\ldots,N\} \setminus \Ss} \lVert \vb^T \G(\x_i) \rVert^2$
        \State Set $\Ss \gets \Ss \cup \{i^*\}$
        \If{$|\Ss|=m$} \State \textbf{break} \EndIf
        \State Set $\widetilde{\G} \gets \Pb \G(\x_{i^*})$
        \State Set $\Pb \gets \Pb - \widetilde{\G}\widetilde{\G}^+$
    \Until{$\tr\left[\Pb\right] < 0.5$}
    \State Set $\w$ to be the uniform design on $\{\x_i: i \in \Ss\}$
\end{algorithmic}
\end{algorithm}

\begin{theorem}[]\label{thm2}
    The algorithm mKYM (Algorithm \ref{aINI}) finds a set $\Ss$ of size at most $m$, such that the uniform design $\w$ supported on $\{\x_i: i \in \Ss\}$ is nonsingular.
\end{theorem}
\bigskip
In the actual implementation of Algorithm \ref{aINI}, $\vb$ in Line 3 is obtained as $\Pb \z$, $\z \sim \mathcal{N}_m(\0_m, \I_m)$, which gives $\vb \in \C(\Pb)$, $\vb \neq \0_m$, with probability one. In case of ties in Line 4, we take the first index maximizing $\lVert \vb^T \G(\x_i) \rVert^2$. The squared norms in Line 4 are computed by multiplying $\vb^T(\G(\x_1), \ldots, \G(\x_N))$, squaring all the elements of the resulting vector and then appropriately splitting the elements into smaller sums.

\subsection{Candidate points}\label{subsec:cand}

For a current design $\w$, we select the most promising candidates for new support points in accord with the directional derivatives of $\Phi_p$, evaluated in $\M(\w)$ and directions $\Hb(\x_1),\ldots,\Hb(\x_N)$. The reason is that they determine the design points that lead to the largest local increase in the optimality criterion $\Phi_p$. Note that the directional derivatives can be computed by explicit formulas (see \eqref{eq:dirderPhip} in Section \ref{sec:problem}):
\begin{equation*}
    \partial \Phi_p \left[\M(\w), \Hb(\x_i)\right] = \Phi_p \left[\M(\w)\right] \left( \frac{(\mathbf{g}_p(\w))_i}{\text{tr}\left[\M^{-p}(\w)\right]} -1 \right),
\end{equation*}
for $i=1,\ldots,N$, where $\Phi_p \left[\M(\w)\right]$ and $\text{tr}\left[\M^{-p}(\w)\right]$ are positive numbers, constant with respect to $i$. That is, the ordering of the design points given by $\partial \Phi_p \left[\M(\w), \Hb(\x_i)\right]$, $i=1,\ldots,N$, is the same as the ordering given by the components of $\mathbf{g}_p(\w)$, as defined in equation \eqref{eq:g}. This explains Line 4 of Algorithm \ref{mREX}. In fact, it is straightforward to verify that this generalizes the choice of the candidate support points in the original REX algorithm for the single-response problem and criteria of $D$- and $A$-optimality. 

In the implementation of mREX, we efficiently compute $\g_p(\w)$ by using the Cholesky decomposition $\M^{-p-1}(\w)=\Lb\Lb^T$. Then instead of \eqref{eq:g}, we can use
\begin{equation}
    (\mathbf{g}_p(\w))_i = \sum_{j=1}^s\sum_{l=1}^m \left[\textbf{G}^T(\x_i)\textbf{L}\right]_{jl}^2\label{gxM_GGT}
\end{equation} 
for $i=1,\ldots,N$. Evaluating the sums of squares on the right-hand-side of \eqref{gxM_GGT} using available numerical matrix procedures can be much faster than computing the trace of $\G^T(\x_i)\M^{-p-1}(\w)\G(\x_i)$, especially because it can be performed in a vectorized manner for all $i=1,\ldots,N$ simultaneously (i.e., by avoiding explicit loops).

\subsection{Optimal exchanges}\label{subsec:opt}

In Line 8 of mREX, we solve the single-response optimization problem of the form 
\begin{equation}\label{ulohaoptkrokvseob}
\alpha^* \in \mathrm{argmax}_{\alpha \in [\ell, u]} \Phi_p\left[\M + \alpha \mathbf{\Delta}\right].
\end{equation}

Here, $\ell:=-w_{i^g_l} \leq 0$ and $u:=w_{i^s_k}>0$ are the boundary points of the optimization interval, $\M=\M(\w)$ is a positive definite $m \times m$ matrix and $\mathbf{\Delta}:= \Hb(\x_{i^g_l})-\Hb(\x_{i^s_k})$ is an $m \times m$ matrix. During computation, the optimization problem \eqref{ulohaoptkrokvseob} is solved many times; therefore, its rapid solution is key to the performance of the algorithm. 

For the single-response case, in which we always have $\rank(\mathbf{\Delta}) \leq 2$, and for the criteria of $D$- and $A$-optimality, there exists an analytic solution of \eqref{ulohaoptkrokvseob}; see \cite{bohning} for $D$-optimality and \cite{rex} for $A$-optimality. Recently, the analytic formulas for $D$-optimality have been extended to the case $\rank(\mathbf{\Delta})=3$; see \cite{ponte2023}. Note that in Section \ref{subsub:Dopt} we show that for $D$-optimality, a closed-form algebraic expression for $\alpha^*$ can in principle be derived if $\rank(\mathbf{\Delta}) \leq 5$. However, the general multi-response case does not seem to lend itself to an analytical solution of \eqref{ulohaoptkrokvseob}, not even for the most common criteria. We therefore examine efficient numerical solutions.

\subsubsection{General Kiefer's criteria}\label{subsubsec:kiefer}

The properties of $\Phi_p$ imply that the objective function
\begin{equation*}
 \phi(\alpha):=\Phi_p\left[\M + \alpha \mathbf{\Delta}\right]
\end{equation*}
of \eqref{ulohaoptkrokvseob} is non-negative, continuous and concave on $[\ell, u]$. The function $\phi$ is positive and smooth inside $(\ell,u)$. However, the derivative $\phi'(\ell\!\to\!u)$ at $\ell$ in the direction of $u$, or the derivative $\phi'(u\!\to\!\ell)$ at $u$ in the direction of $\ell$, can be infinite. 
\bigskip

An important empirical observation is that during the execution of mREX a solution of \eqref{ulohaoptkrokvseob} often corresponds to the full exchange of the weight of the design points with indices $i^g_l$ and $i^s_k$, i.e., the solutions are frequently $\alpha^*=\ell$ or $\alpha^*=u$. Therefore, it is beneficial to rapidly determine whether one of these two cases holds. Concavity of $\phi$ yields that if $\phi'(\ell\!\to\!u) \leq 0$ then $\alpha^*=\ell$ and if $\phi'(u\!\to\!\ell) \leq 0$ then $\alpha^*=u$. If we are able to efficiently verify these conditions, many weight exchanges will be performed rapidly. 

If $\alpha \in (\ell, u)$ then the sign of the standard derivative of $\phi$ in $\alpha$, denoted by $\phi'(\alpha)$, is the same as the sign of the directional derivative (cf. \eqref{eq:dirderPhip})
\begin{equation}\label{eq:dirderaux}
   \partial \Phi_p [\M + \alpha \mathbf{\Delta},  \M +  u\mathbf{\Delta}] =\\
   \Phi_p[\M + \alpha \mathbf{\Delta}]\left(\frac{\mathrm{tr}[(\M + \alpha \mathbf{\Delta})^{-p-1}(\M +  u \mathbf{\Delta})]}{\mathrm{tr}[(\M + \alpha \mathbf{\Delta})^{-p}]} - 1\right).
\end{equation}
Because $\Phi_p[\M + \alpha \mathbf{\Delta}]>0$ and $\mathrm{tr}[(\M + \alpha \mathbf{\Delta})^{-p}]>0$, it is straightforward to verify that the sign of \eqref{eq:dirderaux} is the same as the sign of $\mathrm{tr}[(\M + \alpha \mathbf{\Delta})^{-p-1}\mathbf{\Delta}]$. Note that for $D$-optimality this result also follows from the Jacobi's formula (see Section 8.3 in \cite{magnusneudecker}).

Clearly, if $\M+\ell\mathbf{\Delta}$ is singular, then $\phi(\ell)=0$, that is, $\alpha^* \neq \ell$. Similarly, if $\M+u\mathbf{\Delta}$ is singular, then $\phi(u)=0$, and $\alpha^* \neq u$. Provided that $\M+\ell \mathbf{\Delta}$ is nonsingular, the sign of $\phi'(\ell\!\to\!u)$ is finite and the same as the sign of $\mathrm{tr}[(\M + \ell \mathbf{\Delta})^{-p-1}\mathbf{\Delta}]$. Similarly, if $\M + u \mathbf{\Delta}$ is nonsingular, then $\phi'(u\!\to\!\ell)$ is finite and its sign is opposite to the sign of $\mathrm{tr}[(\M + u \mathbf{\Delta})^{-p-1}\mathbf{\Delta}]$. Combining these ideas allows us to check whether $\phi'(\ell\!\to\!u) \leq 0$ or $\phi'(u\!\to\!\ell) \leq 0$ and, if some of these conditions holds, determine the $\alpha^*$ as either $\ell$ or $u$.
\bigskip

The approach outlined above requires evaluation of $\phi$ and calculation of directional derivatives at $\ell$ or $u$. In our implementation we opt for an approximate but more efficient numerical approach: We choose a small tolerance $\varepsilon>0$ and 
\begin{enumerate}
    \item Find a numerical estimate $\widetilde{\phi'}(\ell\!\to\!u) = (\phi(\ell + 2\varepsilon) - \phi(\ell + \varepsilon)) / \varepsilon$ of $\phi'(\ell\!\to\!u)$ based on the values of $\phi$ in $\ell+\varepsilon$ and $\ell+2\varepsilon$; note that both $\M+(\ell+\varepsilon)\mathbf{\Delta}$ and $\M+(\ell+2\varepsilon)\mathbf{\Delta}$ are nonsingular, and the evaluation of $\phi(\ell+\varepsilon)$ and $\phi(\ell+2\varepsilon)$ is straightforward.
    \item If $\widetilde{\phi'}(\ell\!\to\!u) \leq 0$ then, due to concavity of $\phi$, $\alpha^* \in [\ell, \ell+2\varepsilon)$, therefore we can set $\alpha^*=\ell$.
    \item If $\widetilde{\phi'}(\ell\!\to\!u)>0$ we find a numerical estimate $\widetilde{\phi'}(u\!\to\!\ell)$ of $\phi'(u\!\to\!\ell)$ based on the values of $\phi$ in $u-\varepsilon$ and $u-2\varepsilon$; similarly as above, note that both $\M+(u-\varepsilon)\mathbf{\Delta}$ and $\M+(u-2\varepsilon)\mathbf{\Delta}$ are nonsingular, therefore the evaluation of $\phi(u-\varepsilon)$ and $\phi(u-2\varepsilon)$ is straightforward.
    \item If $\widetilde{\phi'}(u\!\to\!\ell) \leq 0$  then, due to concavity of $\phi$, $\alpha^* \in (u-2\varepsilon,u]$, therefore we can set $\alpha^*=u$.
    \item If neither of the above conditions is satisfied, we use a one-dimensional optimization procedure to find the maximum of $\phi$ on $[\ell+2\varepsilon, u-2\varepsilon]$. Similarly as above, for each $\alpha \in [\ell+2\varepsilon, u-2\varepsilon]$ the matrix $\M + \alpha \Delta$ is nonsingular, therefore the optimization procedure does not encounter singularities. Moreover, $\phi$ is positive, concave and smooth at the optimization interval. Therefore, we can find the optimum $\alpha^*$ rapidly via a univariate optimization procedure (such as the \texttt{base} function \texttt{optimize} in \texttt{R}).
\end{enumerate}

In the last step above, we can alternatively find $\alpha^*$ as a solution of $\mathrm{tr}[(\M + \alpha^* \mathbf{\Delta})^{-p-1}\mathbf{\Delta}]=0$ over $[\ell+2\varepsilon, u-2\varepsilon]$, utilizing a univariate root-solving algorithm (such the \texttt{base} function \texttt{uniroot} in \texttt{R}). However, our experience suggests that finding the root of the derivative is somewhat less efficient than computing the optimum directly, which we ultimately chose for our implementation.

\subsubsection{D-optimality}\label{subsub:Dopt}

For $D$-optimality, as the most important Kiefer's criterion, we provide a more  efficient approach than the general one from Section \ref{subsubsec:kiefer}.
For this criterion, the problem \eqref{ulohaoptkrokvseob} is equivalent to maximizing the determinant of $\M + \alpha \mathbf{\Delta}$ over $\alpha \in [\ell,u]$. Let $\Q:=-\M^{-1}\mathbf{\Delta}$. The key observation is that for any real $\alpha \neq 0$
\begin{eqnarray*}
    \det\left[\M + \alpha \mathbf{\Delta}\right] &=& \det\left[\M\right]\det\left[\mathbf{I}_m - \alpha \Q\right]\\
    &=& \alpha^m \det\left[\M\right]\det\left[\alpha^{-1}\mathbf{I}_m - \Q\right]=\alpha^m \det\left[\M\right] P_{\Q}(\alpha^{-1}),
\end{eqnarray*}
where
\begin{eqnarray*}
    P_{\Q}(\lambda) &=& c_0 + c_1 \lambda + \cdots + c_m \lambda^m\\ 
    &=& c_{m-r}\lambda^{m-r}+c_{m-r+1}\lambda^{m-r+1}+ \cdots + c_m \lambda^m
\end{eqnarray*}
is the characteristic polynomial of $\Q$ and $r$ is the rank of $\Q$. We obtained
\begin{equation}\label{eq:Dpoly}
    \det\left[\M^{-1}\right]\det\left[\M + \alpha \mathbf{\Delta}\right] = c_{m-r}\alpha^r + c_{m-r+1} \alpha^{r-1} + \cdots + c_m.
\end{equation}
Note that the above equation also holds for $\alpha=0$, since $c_m=1$. Therefore, for $D$-optimality, the solution $\alpha^*$ of the optimization problem $\eqref{ulohaoptkrokvseob}$ is the maximum of the polynomial on the right-hand side of \eqref{eq:Dpoly} over the interval $[\ell,u]$. 

Because the coefficients of $P_\Q$ can be expressed as functions of the elements of $\Q$ and the quartic polynomial equation can be solved by radicals in the general case, it is theoretically possible to write down explicit formulas for the solution of \eqref{ulohaoptkrokvseob} under $D$-optimality, provided that $r \leq 5$. This covers the case of a regression model with a two-dimensional response at each design point. However, beyond the classical case of $r=2$, the formulas become overcomplicated, and probably not worth implementing.

Rather, we propose a simpler and completely general approach: we first compute the coefficients of the characteristic polynomial $P_{\Q}$ and then numerically maximize the polynomial on the right-hand side of \eqref{eq:Dpoly} over $[\ell,u]$. Numerical maximization of a fixed polynomial is generally much faster and more reliable than repeated evaluation of determinants. Importantly, the characteristic polynomial of a matrix can be computed using the Faddeev-LeVerrier method, implemented for instance as a function \texttt{charpoly} of the \texttt{R} library \texttt{pracma}.

The complexity of the Faddeev-LeVerrier method applied directly to the $m \times m$ matrix $\Q$ is at least $O(m^3)$ (e.g., \cite{bar2021}). Therefore, for larger $m$, the computational cost can be substantial. However, if $m$ is large, we typically have $2s < m$, and we can achieve significant numerical simplification by using the following approach. The generalized matrix determinant lemma (see Theorem 18.1.1 in \cite{harville}) implies for any $\alpha \in [\ell,u]$ and indices $i^g_l$, $i^s_k$
\begin{eqnarray*}
    \det\left[\M + \alpha \mathbf{\Delta} \right] &=& \det\left[\M + \alpha (\G(\x_{i^g_l})\G^T(\x_{i^g_l})-\G(\x_{i^s_k})\G^T(\x_{i^s_k})\right]\\
    &=& \det\left[\M + \alpha(\G(\x_{i^g_l}),\G(\x_{i^s_k}))(\G(\x_{i^g_l}), -\G(\x_{i^s_k}))^T\right]\\
    &=& \det\left[\M\right]\det\left[\I_{2s} + \alpha (\G(\x_{i^g_l}), -\G(\x_{i^s_k}))^T \M^{-1}(\G(\x_{i^g_l}),\G(\x_{i^s_k}))\right].
\end{eqnarray*}
Consequently, to obtain an optimal solution $\alpha^*$ of \eqref{ulohaoptkrokvseob} for $p=0$, we can solve the same optimization problem, but with $\widetilde{\M}=\I_{2s}$ instead of $\M$ and with the $2s \times 2s$ matrix 
\begin{equation*}
    \widetilde{\mathbf{\Delta}} = (\G(\x_{i^g_l}), -\G(\x_{i^s_k}))^T \M^{-1}(\G(\x_{i^g_l}),\G(\x_{i^s_k}))
\end{equation*}
instead of $\mathbf{\Delta}$. This means that for the characteristic polynomial approach, we only need to apply the Faddeev-LeVerrier method to a $2s \times 2s$ matrix; in a typical application, e.g., with a binary response ($s=2$), the computation of characteristic polynomials for such small matrices is extremely rapid, independently of the number $m$ of parameters.

\section{Extensions}\label{sec:further}

This section discusses extensions of the optimal design problem that the mREX algorithm can directly address.

\subsection{Variable response dimension and variable covariance matrices}

While the linear regression model in \eqref{eq:model_multi} assumes observations with fixed dimension $s$ and a constant covariance matrix, the mREX algorithm can be applied to scenarios where both the response dimensions and response covariance matrices vary (cf. Section II.5.3 in \cite{pazman86}). More formally, we can consider the model
\begin{equation}\label{eq:multires}
    \y_\x = \F^T(\x) \boldsymbol{\beta} + \boldsymbol{\epsilon}_\x, 
\end{equation}
where $\y_\x$ is the $s(\x)$-dimensional response at the design point $\x \in \mathcal{X}$, $\boldsymbol{\beta} \in \mathbb{R}^m$ is the vector of parameters, and $\F(\x)$ is a known matrix of size $m \times s(\x)$. The vectors $\boldsymbol{\epsilon}_\x$ are assumed to be $\mathcal{N}_{s(\x)}(\mathbf{0}_{s(\x)}, \sigma^2\boldsymbol{\Sigma}(\x))$ distributed and uncorrelated across different trials, with positive definite covariance matrices $\boldsymbol{\Sigma}(\x)$ of dimension $s(\x) \times s(\x)$. In this model, the information matrix corresponding to a design $\w$ is 
\begin{equation*}
    \M(\w) = \sum_{i=1}^N w_i \F(\x_i) \boldsymbol{\Sigma}^{-1}(\x_i) \F^T(\x_i) = \sum_{i=1}^N w_i \Hb(\x_i),
\end{equation*}
where the elementary information matrix of the design point $\x_i \in \mathcal{X}$ is 
\begin{equation*}
        \Hb(\x_i) =  \G(\x_i)\textbf{G}^T(\x_i), \:\: \G(\x_i) = \F(\x_i)\boldsymbol{\Sigma}^{-1/2}(\x_i).
\end{equation*}

\subsection{Nonlinear and generalized linear regression models}\label{subsec:nonlinear}

Our computational approach also naturally extends to locally optimal designs for nonlinear multi-response regression models of the form
\begin{equation}\label{eq:model_multi2}
    \y_\x = \F(\x, \boldsymbol{\beta}) + \boldsymbol{\epsilon}_\x,
\end{equation}
where $\F:\mathcal{X} \times B \to \R^s$ is a nonlinear mean-value function, differentiable in the interior $B^\circ \neq \emptyset$ of $B \subseteq \R^m$, and all other aspects are the same as for the standard model \eqref{eq:model_multi}.

The optimal design in nonlinear models typically depends on the value of the unknown parameters, and a nominal value $\boldsymbol{\beta}^* \in B^\circ$ for the unknown parameters is usually considered. In this case, if the function $\F$ satisfies some general regularity assumptions, the approach of locally optimal design leads to elementary information matrices $\Hb(\x_i, \boldsymbol{\beta}^*)$ and matrices $\textbf{G}(\x_i,\boldsymbol{\beta}^*)$ of the form
\begin{equation}\label{eq:EIM_nonlinear}
        \Hb(\x_i, \boldsymbol{\beta}^*) =  \textbf{G}(\x_i,\boldsymbol{\beta}^*)\textbf{G}^T(\x_i,\boldsymbol{\beta}^*), \:\: \textbf{G}(\x_i,\boldsymbol{\beta}^*) = \frac{\partial \F(\x_i, \boldsymbol{\beta}^*)}{\partial \boldsymbol{\beta}} \boldsymbol{\Sigma}^{-1/2},
\end{equation}
and the corresponding $\Phi$-optimal designs are called local $\Phi$-optimal designs (see \cite{chernoff}, or \cite{pronzatopazman2013}, Chapter 5).
\bigskip

The approach of local optimality can similarly be used for generalized linear models (e.g., \cite{khuri}, \cite{AtkinsonWoods}). The expected value vector of $\y_\x$ in a generalized linear model is $(g_1^{-1}(\h_1^T(\x)\boldsymbol{\beta}), \ldots, g_s^{-1}(\h_s^T(\x)\boldsymbol{\beta}))^T$ for strictly monotone differentiable link functions $g_1, \ldots, g_s$ and known regressors $\h_1(\x), \ldots, \h_s(\x)$. For a nominal value $\boldsymbol{\beta}^*$, we have
$$
        \Hb(\x_i, \boldsymbol{\beta}^*) =  \textbf{G}(\x_i,\boldsymbol{\beta}^*)\textbf{G}^T(\x_i,\boldsymbol{\beta}^*),
$$
where
$$
        \textbf{G}(\x_i,\boldsymbol{\beta}^*)  = \begin{bmatrix}
        \sqrt{v_1(\x_i, \boldsymbol{\beta}^*)}\h_1(\x_i), & \ldots, & 
        \sqrt{v_s(\x_i, \boldsymbol{\beta}^*)}\h_s(\x_i)
        \end{bmatrix} \boldsymbol{\Sigma}^{-1/2},
$$
for known functions $v_1, \ldots, v_s$ that correspond to the distributions of the elements of $\y_\x$ and to the links $g_1, \ldots, g_s$ (for some common functions $v$, see, e.g., \cite{gex}).

\section{Application: dose-response studies}\label{sec:examples}

We demonstrate the usefulness of our algorithm and compare its performance to competitors by applying it to a problem related to designing dose-finding studies. All calculations were performed on a computer with a 64-bit Windows 11 operating system running an AMD Ryzen 7 5800H CPU processor at $3.20$ GHz with $16$ GB RAM. All codes in software \texttt{R} are available upon request from the authors.

The primary objective of dose-response studies is to establish a balance between the therapeutic benefits of a drug and the potential risks associated with different dosage levels (see, e.g., \cite{ting2006}, \cite{fedorovleonov2013}, \cite{magnusdottir2016}), that is, we often measure a bivariate response corresponding to a level of efficacy and toxicity. However, the multi-response regression can be used to model any collection of possibly correlated measurements as a multivariate response to a given dose of a drug, not only responses corresponding to efficacy and toxicity.

\subsection{Bivariate Emax model}\label{ssEmax}
We first consider a bivariate Emax model, as in \cite{magnusdottir2016}:
\begin{equation}
			\begin{bmatrix}
				y_{x,1} \\
				y_{x,2}
			\end{bmatrix}
			= 
			\begin{bmatrix}
			f_1(x, \boldsymbol{\beta}_1) \\
			f_2(x, \boldsymbol{\beta}_2)
			\end{bmatrix}
			+
			\begin{bmatrix}
				\varepsilon_{x,1} \\
				\varepsilon_{x,2}
			\end{bmatrix}, \label{multivar_emax}
\end{equation}
where 
\begin{equation}\label{eEmaxF}
f_i(x, \boldsymbol{\beta}_i) = E_{0,i} + \frac{x \cdot E_{\mathrm{max},i}}{x + ED_{50,i}},
\end{equation}
$\boldsymbol{\beta}_i = (E_{0,i}, E_{\mathrm{max},i}, ED_{50,i})^T$ for $i = 1,2$, and the vector of all parameters is $\boldsymbol{\beta}=(\boldsymbol{\beta}^T_1,\boldsymbol{\beta}^T_2)^T$.
The design variable $x \in \left[a,b \right]$ represents the dose. The parameters $E_{0,i}$ are the placebo effects when the dose $x$ is zero and the parameters $E_{\mathrm{max},i}$ represent by how much the maximal achievable effects of the drug exceed the placebo effects. Furthermore, the expected response value at $x = ED_{50,i}$ is $f_i(ED_{50,i}, \boldsymbol{\beta}_i) = E_{0,i} + \frac{ E_{\mathrm{max},i}}{2}$, so the parameters $ED_{50,i}$ represent the dose levels at which the effects correspond to the average of the placebo effects and maximal achievable effects. As is commonly assumed (see \cite{magnusdottir2016}, \cite{schorning2017} or \cite{tsirpitzi_miller}), the covariance matrix $\boldsymbol{\Sigma}$ of $(\varepsilon_{x,1},\varepsilon_{x,2})^T$ is considered to be known. As \eqref{multivar_emax} is a nonlinear regression model, the local optimality approach leads to elementary information matrices given by \eqref{eq:EIM_nonlinear}, where $\F(x_i, \boldsymbol{\beta}^*) = \left(f_1(x_i, \boldsymbol{\beta}_1^*), f_2(x_i, \boldsymbol{\beta}_2^*)\right)^T$.
\bigskip

Local $D$-optimal minimally supported design for the Emax model \eqref{multivar_emax} over the continuous interval $[a,b]$ is given by \cite{schorning2017}. Here, ``$D$-optimal minimally supported design'' means that the design is $D$-optimal among all designs with the support size of $K=3$ points, which is the smallest support size of any nonsingular design. Note that such designs need not be $D$-optimal among all designs: it may happen that the overall $D$-optimal design is supported on more than three points. The design given by \cite{schorning2017} places weight $1/3$ on the following three points: $a$, $b$ and $x_M$, where 
\begin{equation}
    x_M = \frac{\sqrt{(a+ED_{50,1}^*)(a+ED_{50,2}^*)(b+ED_{50,1}^*)(b+ED_{50,2}^*)}+ab-ED_{50,1}^*ED_{50,2}^*}{a+b+ED_{50,1}^*ED_{50,2}^*}, \label{xM}
\end{equation}
and $ED_{50,1}^*$ and $ED_{50,2}^*$ denote the nominal values of the corresponding model parameters.

We consider the situation that $x \in \left[a,b\right] = \left[0,500 \right]$ and the nominal parameters are $E_{0,1}^* = E_{0,2}^* = 60, ED_{50, 1}^* = ED_{50, 2}^* = 25$ and $E_{\mathrm{max}, 1}^* = E_{\mathrm{max}, 2}^* = 294$. These parameter values correspond to the efficacy parameter values in the Emax dose-finding trial involving an anti-asthmatic drug presented in \cite{bretz2010}. The covariance matrix $\boldsymbol{\Sigma}$ used in the computations is 
\begin{equation}\label{eSigma}
\boldsymbol{\Sigma} = \begin{pmatrix}
  1 & 1/2 \\
  1/2 & 1
\end{pmatrix}.
\end{equation}
The $D$-optimal minimally supported design given by \eqref{xM} is then
\begin{equation}
\w_{0}^* =
\begin{pmatrix}
0 & 22.727 & 500 \\
\frac{1}{3} & \frac{1}{3} & \frac{1}{3} \\
\end{pmatrix}. \label{fixed_kirsten}
\end{equation}

The design $\w_{0}^*$ was constructed under several assumptions: specific values for the nominal parameters were used, the criterion was fixed as $D$-optimality, and the design is $D$-optimal only within the class of minimally supported designs. We therefore examine the quality of $\w_{0}^*$ when some of these settings change, which can be done by employing the mREX algorithm. Specifically, we perform a sensitivity study on the efficiency of $\w_{0}^*$: we examine its efficiency relative to the $D$-optimal design for various other values of the nominal parameters, and relative to the $\Phi_p$-optimal designs for other $p$. Note that the analytic results given by \cite{schorning2017} cannot be used for such analysis, as they do not cover actually $D$-optimal designs (only $D$-optimal minimally supported designs), nor other optimality criteria. 

\paragraph{Nominal parameters}
We change the value of the nominal parameter $ED_{50,2}^* = 25$ to new values $ED_{50,2}^* = e \in \left[ 5,490 \right]$, while keeping the other nominal parameters fixed. For each $e \in \{ 5, 10, 15, \dots, 490 \}$, we compute the $D$-optimal design $\vb^*_e$ using the mREX algorithm, and the $D$-optimal minimally supported design $\vb_{\min, e}^*$ via the analytic expression \eqref{xM}. For each value of $e$, we compute the efficiency of the original design $\w_0^*$ (which is the same as $\vb_{\min, 25}^*$) relative to $\vb_e^*$ and the efficiency of $\vb_{\min, e}^*$ relative to $\vb_e^*$; the results are given in Figure \ref{fig:minsuppdesigns} (left). The latter efficiency calculation allows us to determine for which values of $ED_{50,2}^*=e$ is the $D$-optimal minimally supported design $\vb_{\min, e}^*$ actually $D$-optimal among all designs. The figure shows that the efficiency of the original design $\w_0^*$ decreases slightly with changing $ED_{50,2}^*$, but still remains reasonably large (above 70\%). Figure \ref{fig:minsuppdesigns} also illustrates that for $ED_{50,2}^*$ greater than approximately $100$, the efficiency of $\vb_{\min, e}^*$ is lower than 100\%, indicating that the locally $D$-optimal design is no longer supported on only three points for larger values of $ED_{50,2}^*$. In contrast, the results show that the $D$-optimal minimally supported design $\vb_{\min, 25}^*$ is actually $D$-optimal among all designs for $ED_{50,2}^*=25$.

\paragraph{Other criteria} We also change $p$ within the interval $[0,6]$ discretized by $0.1$, and calculate the efficiency of the $D$-optimal (minimally supported) design $\w_0^*$ relative to the $\Phi_p$-optimal design computed by mREX; see Figure \ref{fig:minsuppdesigns} (the right panel). Similarly to the case of changing $ED_{50,2}^*$, the efficiency of the $D$-optimal design decreases only moderately (staying above 70\%). Note that other published algorithms cannot reproduce these results without significant extensions, as they are only formulated for (all or some) integer values of $p$.
\bigskip

Overall, our sensitivity study shows that the $D$-optimal minimally supported design \eqref{fixed_kirsten} is actually $D$-optimal among all designs for $ED_{50, 2}^* = 25$. Additionally, it proves that the design is a robust choice for the experiment provided that we suspect deviations of the true $ED_{50, 2}^*$ from its nominal value $25$, and that the design is reasonably efficient not only for $D$-optimality, but with respect to a large range of $\Phi_p$-criteria, which capture different statistical aspects of the parameter estimator. 

\begin{figure}
\includegraphics[width=\linewidth]{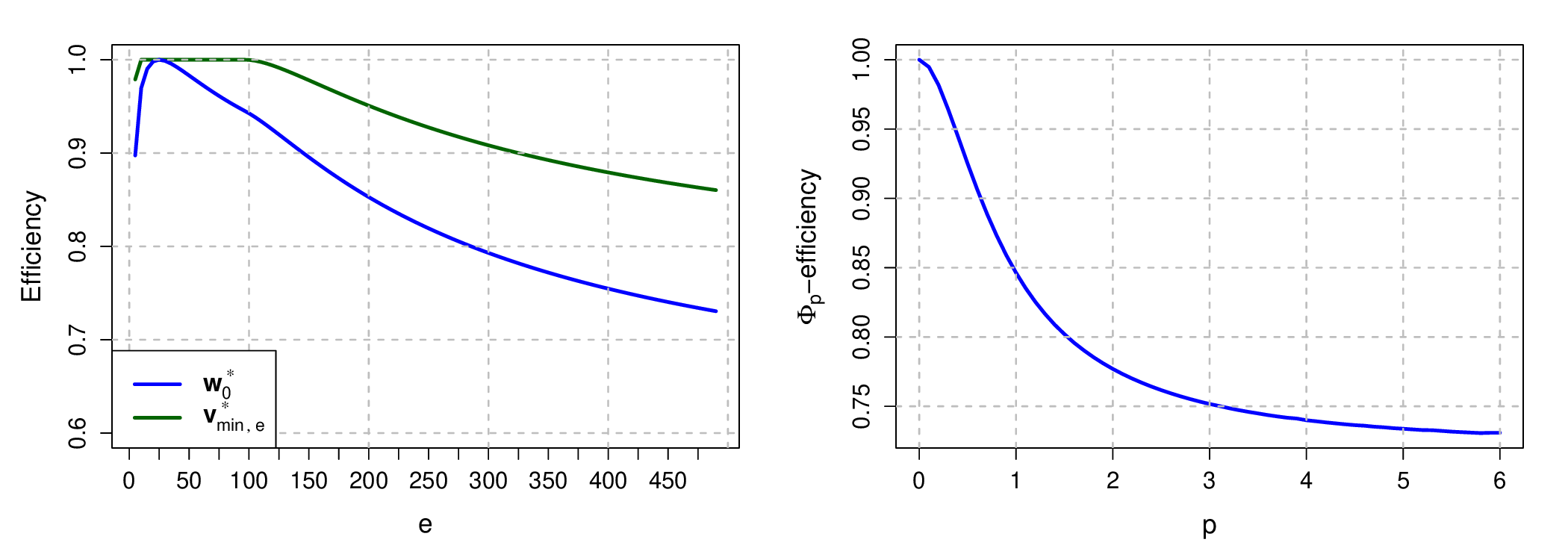}
\caption{Left: The efficiency of the fixed minimally supported design $\w_{0}^*$ and the changing minimally supported designs $\vb_{\min, e}^*$ relative to the $D$-optimal designs $\vb_e^*$ obtained by the mREX algorithm, for changing nominal parameter value of $ED_{50,2}^* = e$. Right: The efficiency of the $D$-optimal minimally-supported design $\w_{0}^*$ relative to the $\Phi_p$-optimal designs, $p \in \left[ 0,6\right]$, obtained by the mREX algorithm.}\label{fig:minsuppdesigns}
\end{figure}

\subsection{Bivariate Emax models with covariates}

To compare the performance of mREX with the current state-of-the-art algorithms, we consider more complex Emax settings. In particular, we suppose that the drug response may vary with patient characteristics (such as age and gender), which can be expressed by including linear effects of covariates in the model. This leads to the bivariate Emax model with covariate effects
\begin{equation}
    y_{x,\textbf{z},i} = E_{0,i} + \frac{x \cdot E_{\mathrm{max},i}}{x+ED_{50,i}} + \boldsymbol{\theta}_i^T \textbf{z} + \varepsilon_{x, \textbf{z},i}, \label{y_covar}
\end{equation}
where $\textbf{z} \in \mathcal{Z} \subseteq \R^k$ are the covariates characterizing the patient receiving the dose $x$, $\mathcal{Z}$ is the set of all permissible covariates, $\boldsymbol{\theta}_i \in \mathbb{R}^k$ is the corresponding vector of parameters, and $i \in \{1,2\}$ is the component of the response vector. Unlike for \eqref{multivar_emax}, the design points for \eqref{y_covar} are multidimensional: $\x = (x, \z^T)^T \in \R^{k+1}$. The total number of parameters in this model is $m = 6+2k$. Model \eqref{y_covar} is presented in \cite{covariates} for a single response, however, we consider bivariate responses. The dosage levels are discretized equidistantly on the interval $\left[0,500\right]$ with $N_d$ discrete values. The nominal values of the parameters are set as in Section \ref{ssEmax}: $E_{0,1}^* = E_{0,2}^* = 60$, $ED_{50,1}^* = ED_{50,2}^*= 25$ and $E_{\mathrm{max},1}^* = E_{\mathrm{max},2}^* = 294$, and the covariance matrix $\boldsymbol{\Sigma}$ is again given by \eqref{eSigma}.

Specifically, we compare the performance of mREX with the multiplicative algorithm (\cite{yu2010}) and the YBT algorithm (\cite{ybt}). Among other algorithms for multi-response problems, the vertex direction method, although simple and versatile, tends to be slow for all but the smallest problems. Moreover, the semidefinite programming and second-order cone programming approaches, while possessing some unique advantages such as the capability to include advanced constraints, can, with the current state of numerical solution methods for these classes of problems, only be used for mid-sized problems and they require special solvers. Therefore, we do not include them in the comparisons.

We examine the model \eqref{y_covar} with no covariates (which is the model \eqref{multivar_emax}) and  with $k \in \{3, 5, 9\}$ covariates. The permissible range for each covariate is $[-1,1]$, discretized to $N_c$ values, with various values of $N_c$. The total number of design points for a given model is therefore $N=N_d \times  N_c^k$. For each model, we computed designs with guaranteed $D$-efficiency of at least $0.99999$ by running the mREX and the YBT algorithms 100 times, and the multiplicative algorithm once\footnote{The multiplicative algorithm is deterministic, therefore each run of this algorithm produces the same result in approximately the same time.}; the summary of the running times is reported in Table \ref{tblTimes}. The table also lists all considered models (i.e., the values of $k$, $N_d$, $N_c$ and the resulting $N=N_d \times N_c^k$). The running times of mREX and YBT are expressed by medians and 5\%- and 95\%-quantiles of the $100$ runs.

Several representative runs of the algorithms are depicted in Figure \ref{fig:Runs}. Both Table \ref{tblTimes} and Figure \ref{fig:Runs} demonstrate that the mREX algorithm outperforms its competitors, and the (relative) difference in running times becomes even more pronounced as the model complexity increases. Moreover, as can be seen in Figure \ref{fig:Runs}, the mREX algorithm generally achieves a better design than its competitors at each time point during its run; one exception is the very beginning of the run, where it takes mREX longer to find the first design due to its slightly more time-consuming, but significantly better-performing initiation method. We also see that the multiplicative algorithm is not competitive in providing designs with a very high efficiency, while the YBT algorithm occasionally exhibits sudden jumps to near-optimal solutions (the ``nearly vertical'' lines in some cases in Figure \ref{fig:Runs}) after multiple iterations with only modest improvements. Nevertheless, even this sudden improvement is generally not enough to make the YBT algorithm faster than the mREX algorithm. Note, however, that the actual running time comparisons of different algorithms should always be taken with a grain of salt, as they are highly dependent on factors such as software efficiency, hardware specifications, and implementation details.

\begin{table}[ht!]
\caption{Performance of the mREX, YBT and multiplicative (MUL) algorithms in computing a design with $D$-efficiency at least 0.99999 for the model \eqref{y_covar}. For mREX and YBT, the tabulated values are ``median (5\%-quantile, 95\%-quantile)'' of 100 runs of the algorithms; for MUL, the running time of one run is reported (if smaller than 100 seconds). All times are in seconds.}\label{tblTimes}\begin{tabular}{cccc}
\toprule
Model \eqref{y_covar} & mREX & YBT  & MUL \\ \midrule
 $k=0$, $N = N_d = 50 001$ & $0.050$  $(0.04, 0.09)$ & $0.11$  $(0.09, 0.16)$ & $> 100$ \\ 
 $k=0$, $N = N_d = 500 001$ & $0.15$  $(0.14, 0.21)$ & $0.42$  $(0.38, 0.54)$ & $> 100$ \\ 
 $k=3$, $N_d = 26$, $N_c = 3$, $N=702$ & $0.22$  $(0.17, 0.28)$ & $0.44$  $(0.36, 0.54)$ & $1.75$ \\ 
 $k=3$, $N_d = 26$, $N_c = 9$, $N=18954$ & $0.25$  $(0.19, 0.31)$ & $0.48$  $(0.41, 0.57)$ & $31.27$ \\ 
 $k=5$, $N_d = 26$, $N_c = 3$, $N=6318$ & $1.17$  $(0.96, 1.49)$ & $1.96$  $(1.22, 2.58)$ & $30.25$ \\ 
 $k=5$, $N_d = 26$, $N_c = 7$, $N=436982$ & $2.61$  $(2.27, 3.08)$ & $4.49$  $(3.37, 5.93)$ & $> 100$ \\ 
 $k=9$, $N_d = 26$, $N_c = 2$ $N = 13312$ & $8.85$  $(7.22, 11.06)$ & $45.01$  $(35.30, 58.26)$ & $> 100$ \\ 
 $k=9$, $N_d = 26$, $N_c = 3$, $N = 511758$ & $14.39$  $(12.02, 17.79)$ & $47.58$  $(44.71, 54.35)$ & $> 100$ \\ \botrule
\end{tabular}

\end{table}

\begin{figure}
\includegraphics[width=\linewidth]{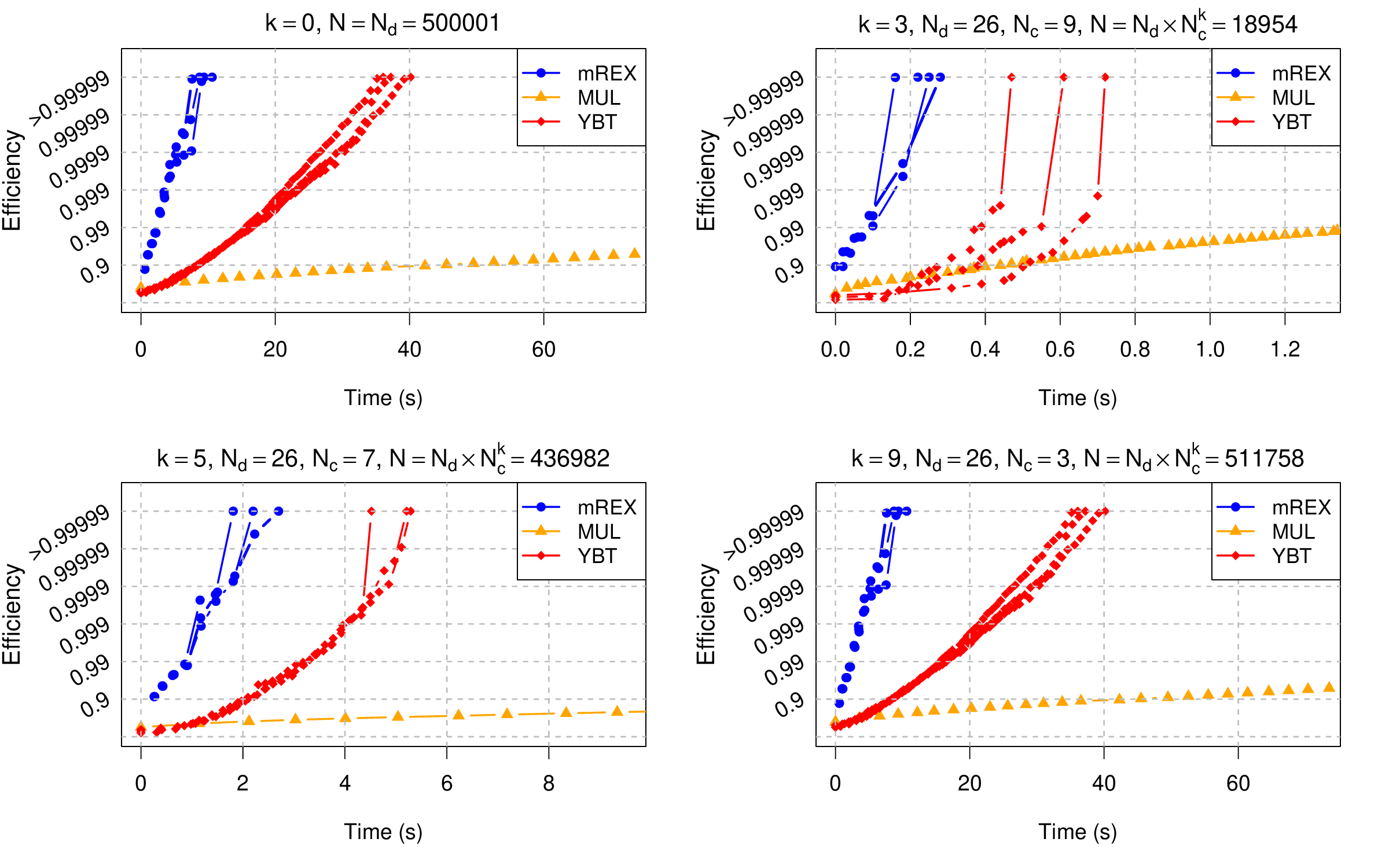}
\caption{Performance of the mREX algorithm for $D$-optimality (blue) in comparison with the classical multiplicative (orange) and the YBT (red) algorithms in the bivariate Emax model \eqref{y_covar} with $k = 0,3,5$ and $9$ covariates, and number of parameters $m = 6 + 2k$. Each line represents one run of the algorithm (i.e., achieved efficiency as a function of the running time). Note that the vertical axis is in logarithmic scale.}\label{fig:Runs}
\end{figure}

\section{Discussion}\label{sDiscussion}

For single-response regression models with finite design spaces, it has been shown (see \cite{rex}) that the REX algorithm exhibits  performance comparable to or better than other optimal design algorithms. However, the use of models with multiple responses is increasingly common in practical applications, and the literature on computing optimal designs for multi-response models remains sparse.

In this paper, we have shown that the proposed mREX algorithm retains the beneficial properties of REX in the context of multi-response regression models. In addition to computational efficiency, mREX offers several advantages: it can handle large problems and is straightforward to implement, not requiring advanced mathematical programming solvers.

A direct generalization of the original REX algorithm to the multi-response setting would not achieve such good performance. The effectiveness of mREX is ensured by two nontrivial improvements: the construction of efficient sparse initial designs via the mKYM method and the rapid computation of optimal design weight exchanges, particularly using the characteristic polynomial approach for $D$-optimality. These novel methods may have broader applications in optimal design algorithms, presenting an interesting direction for future research. 

Unlike REX, the algorithm mREX is formulated for the entire class of differentiable Kiefer's criteria, providing greater flexibility. Note however, that it is straightforward to extend mREX to cover any other optimality criterion for which directional derivatives can be efficiently computed.

The primary limitation of mREX is the absence of theoretical convergence proofs. However, this limitation has minimal practical impact, because we have a theorem to verify the optimality or near-optimality of any design produced by mREX; moreover, in numerous tests across a wide range of problems, mREX consistently converged to the optimal design.

\section*{Declarations}
\bmhead{Competing interests}  The authors have no relevant financial or non-financial interests to disclose.
\bmhead{Funding}  This work was supported by the Collegium Talentum Programme of Hungary and the Slovak Scientific Grant Agency (grant VEGA 1/0327/25).
\bmhead{Author contributions} All authors contributed to the methodology, mathematical proofs, software, writing and reviewing of the article.

\begin{appendices}

\section{Proof of Theorem 1}\label{ap:mineff}

\begin{proof}[Proof of Theorem~{\upshape\ref{thm1}}]
    Let $p \in [0, \infty), \M \in \Ss^m_{++}$ and $\N \in \Ss^m_+$. Define
    \begin{equation*}
        \Psi_{p,\M, \mathbf{N}}(\beta) := \left(\Phi_p\left[(1-\beta)\M + \beta \mathbf{N} \right] - \Phi_p\left[\M \right]\right)/\beta \text{ for } \beta \in (0,1].
    \end{equation*}

    Using concavity of $\Phi_p$ it is straightforward to prove that $\Psi_{p,\M, \mathbf{N}}(\beta)$ is non-increasing for $\beta \in (0,1]$, which implies that $\partial \Phi_p \left[\M, \mathbf{N}\right] \geq \Psi_{p,\M, \mathbf{N}}(\beta)$ for any $\beta \in (0,1]$. In particular, for $\beta=1$, $\M=\M(\w)$, $\mathbf{N}=\M(\w^*)$, where $\w$ is any nonsingular design and $\w^*$ is a $\Phi_p$-optimal design, we have
    \begin{equation*}
    \partial \Phi_p \left[\M(\w), \M(\w^*)\right] \geq \Phi_p\left[\M(\w^*) \right] - \Phi_p\left[\M(\w) \right].
    \end{equation*}
    This above inequality, the definition of the design efficiency and the formula \eqref{eq:dirderPhip} yield
    \begin{equation}\label{eq:auxeff}
     \mathrm{eff}_{\Phi_p}(\w) \geq \frac{\text{tr}\left[\M^{-p}(\w)\right]}{\text{tr}\left[\M^{-p-1}(\w)\M(\w^*)\right]}.
    \end{equation}
    The denominator on the right-side of \eqref{eq:auxeff} can be bounded from above as follows:
    \begin{eqnarray*}
    \text{tr}\left[\M^{-p-1}(\w)\M(\w^*)\right]&=&\sum_{i=1}^N w^*_i \text{tr}\left[\M^{-p-1}(\w)\Hb(\x_i)\right] \\
    &\leq& \max_{i=1,\ldots,N} \text{tr}\left[\M^{-p-1}(\w)\Hb(\x_i)\right].
    \end{eqnarray*}
    Therefore, we conclude that
    \begin{equation*}
      \mathrm{eff}_{\Phi_p}(\w) \geq \frac{\text{tr}\left[\M^{-p}(\w)\right]}{\max_{i=1,\ldots,N} (\mathbf{g}_p(\w))_i}.
    \end{equation*}
\end{proof}

\section{Proof of Theorem 2}\label{ap:ini_reg}

\begin{proof}[Proof of Theorem~{\upshape\ref{thm2}}]
    Let $J \in \{1,\ldots,m\}$ be the number of passes through the repeat-until cycle of Algorithm \ref{aINI}. Let $\Ss_0:=\emptyset$ and $\Pb_0:=\I_m$. For $j=1,\ldots,J$, let $\vb_j$ be the vector $\vb$, $i^*_j$ be the index $i^*$, $\Ss_j$ be the set $\Ss$, $\widetilde{\G}_j$ be the matrix $\widetilde{\G}$, and $\Pb_j$ be the matrix $\Pb$, obtained in the $j$th pass of the repeat-until cycle (with $\Pb_m=\0_{m \times m}$ if $J=m$, i.e., if the break in Line 7 was realized).

    First, we will use induction to show the following claim: for any $j=0,\ldots,J$, the matrix $\Pb_j$ is the orthogonal projector onto $\C^\perp(\M(\Ss_j))$. For $j=0$, the claim is trivial; let the claim hold for some $j \in \{0,\ldots,J-1\}$. Note that $\C(\M(\Ss_j)) = \C(\G_{i^*_1},\ldots,\G_{i^*_j})$. It is simple to verify that $\C(\M(\Ss_{j+1}))$, which is equal to $\C(\M(\Ss_j), \G_{i^*_{j+1}})$, is the direct sum of the mutually orthogonal linear spaces $\C(\M(\Ss_j))$ and $\C(\widetilde{\G}_{j+1})$. Therefore, $\Pb_{\C(\M(\Ss_{j+1}))} = \Pb_{\C(\M(\Ss_j))} + \Pb_{\C(\widetilde{\G}_{j+1})}$, and, as $\Pb_{\C^\perp(\A)} = \I_m - \Pb_{\C(\A)}$ for any matrix $\A$ with $m$ rows,
    \begin{equation}\label{eq:projsum}
        \Pb_{\C^\perp(\M(\Ss_{j+1}))} = \Pb_{\C^\perp(\M(\Ss_j))} - \Pb_{\C(\widetilde{\G}_{j+1})}.
    \end{equation}
    We obtained
    \begin{eqnarray*}
    \Pb_{j+1} &=& \Pb_j - \widetilde{\G}_{j+1}\widetilde{\G}_{j+1}^+ 
    = \Pb_j - \Pb_{\C(\widetilde{\G}_{j+1})}\\
    &=& \Pb_{\C^\perp(\M(\Ss_j))} - \Pb_{\C(\widetilde{\G}_{j+1})}
    = \Pb_{\C^\perp(\M(\Ss_{j+1}))},
    \end{eqnarray*}
    where the first equality corresponds to Line 10 of Algorithm \ref{aINI}, the second equality is the standard property of an orthogonal projector, the third equality is based on our induction assumption, and the last equality follows from \eqref{eq:projsum}. Consequently, we proved the claim for $j+1$. Thus, $\Pb_j$ is the orthogonal projector onto $\C^\perp(\M(\Ss_j))$ for any $j \in \{0,1,\ldots,J\}$.
    
    We will now show that $\tr\left[\Pb_{j+1}\right] \leq \tr\left[\Pb_j\right] - 1$ for all $j=0,1,\ldots,J-1$ which, in light of the claim proved above is equivalent to $\tr[\Pb_{\C^\perp(\M(\Ss_{j+1}))}] \leq \tr[\Pb_{\C^\perp(\M(\Ss_j))}] - 1$ for all $j=0,1,\ldots,J-1$. Taking \eqref{eq:projsum} into account, and the fact that the trace of a projector is equal to its rank, we only need to show that $\rank[\Pb_{\C(\widetilde{\G}_{j+1})}] \geq 1$ or, equivalently, that $\widetilde{\G}_{j+1}$ is not zero. 
    
    Assume that $\widetilde{\G}_{j+1}=\Pb_j\G_{i^*_{j+1}}$ \emph{is} a zero matrix. Then $\C(\G_{i^*_{j+1}}) \subseteq \C^\perp(\Pb_j)$ and, since $\vb_{j+1} \in \C(\Pb_j)$, we have $\|\vb^T_{j+1} \G_{i^*_{j+1}}\|^2=0$. Because of the norm-maximizing selection of $\G_{i^*_{j+1}}$ in Line 4 of Algorithm \ref{aINI}, we must have $\|\vb^T_{j+1} \G(\x_i)\|^2=0$ for all $i \in \{1,\ldots,N\}$. But then $\C(\G(\x_i)) \subseteq \C^\perp(\vb_{j+1}) \subsetneq \R^m$ for all $i \in \{1,\ldots,N\}$, which violates the nontriviality assumption \eqref{eq:basic}.

    Because $\tr\left[\Pb_{j+1}\right] \leq \tr\left[\Pb_j\right] - 1$ for each $j$, we have $\tr[\Pb_j] = 0$ after at most $m$ steps (i.e., $\tr\left[\Pb_J\right] = 0$ and $J \leq m$), which yields that $\Pb_J = \0_{m \times m}$. Consequently, $\C(\M(\Ss_J)) = \C^\perp(\Pb_J) = \R^m$, implying the nonsingularity of the uniform design $\w$ on $\{\x_i: i \in \Ss_J\} =  \{\x_i: i \in \Ss\}$.
\end{proof}



\end{appendices}


\bibliography{sn-bibliography}

\end{document}